\documentclass[a4paper,USenglish,cleveref, autoref, thm-restate]{lipics-v2021}
\pdfoutput=1 %uncomment to ensure pdflatex processing (mandatatory e.g. to submit to arXiv)
\title{Theoretical insights and an experimental comparison of tango trees and multi-splay trees} %TODO Please add

\titlerunning{Experimental tango and multi-splay trees} %TODO optional, please use if title is longer than one line

\author{Khaleel Al-Adhami}{Massachusetts Institute of Technology, Cambridge, MA, USA 
\and \url{https://adhami.me} 
}{adhami@mit.edu}{https://orcid.org/0009-0004-1362-221X}{}%TODO mandatory, please use full name; only 1 author per \author macro; first two parameters are mandatory, other parameters can be empty. Please provide at least the name of the affiliation and the country. The full address is optional. Use additional curly braces to indicate the correct name splitting when the last name consists of multiple name parts.

\author{Dev Chheda\footnote{corresponding author}}{Massachusetts Institute of Technology, Cambridge, MA, USA
\and \url{https://vedadehhc.github.io} 
}{dchheda@mit.edu}{https://orcid.org/0009-0007-1178-2940}{}

\authorrunning{K. Al-Adhami and D. Chheda} %TODO mandatory. First: Use abbreviated first/middle names. Second (only in severe cases): Use first author plus 'et al.'

\Copyright{Khaleel Al-Adhami and Dev Chheda} %TODO mandatory, please use full first names. LIPIcs license is "CC-BY";  http://creativecommons.org/licenses/by/3.0/

 %TODO mandatory: Please choose ACM 2012 classifications from https://dl.acm.org/ccs/ccs_flat.cfm 

 %TODO mandatory; please add comma-separated list of keywords

\category{} %optional, e.g. invited paper

\relatedversion{} %optional, e.g. full version hosted on arXiv, HAL, or other respository/website
\supplementdetails[subcategory={Source Code}, cite=github]{Software}{https://github.com/adhami3310/tango} %linktext, cite, and subcategory are optional

\acknowledgements{We thank Erik Demaine for answering our queries and providing clarifications regarding the tango tree paper \cite{tango}. We also thank David Karger and Julian Shun for their invaluable feedback on this paper. Finally, we thank Gopal Goel for his assistance with some of the analysis in Section \ref{sec:ms-wset}.}%optional

\nolinenumbers %uncomment to disable line numbering

\newcommand{\IB}{\mathrm{IB}}
\newcommand{\OPT}{\mathrm{OPT}}
\newcommand{\AT}{\mathrm{AT}}
\usepackage{asymptote}

\begin{document}
\maketitle

%TODO mandatory: add short abstract of the document
\begin{abstract}
The tango tree is the first proven $O(\lg \lg n)$-competitive binary search tree (BST). We present the first ever experimental implementation of tango trees and compare the running time of the tango tree with the multi-splay tree and the splay tree on a variety of families of access sequences. We construct access sequences that are intended to test specific properties of BSTs. The results of the other experiments demonstrate the optimality of the splay tree and multi-splay tree on these accesses, while simultaneously demonstrating the tango trees inability to achieve optimality. We prove that the running time of tango trees on the sequential access is $\Theta(n \lg \lg n)$, which provides insight into why the $\Theta(\lg \lg n)$ slow down exists on many access sequences. Motivated by experimental results, we conduct a deeper analysis of the working set access on multi-splay trees, leading to new insights about multi-splay tree behavior. Finally, all of the experiments also reveal insights about large constants and lower order terms in the multi-splay tree, which make it less practical than the splay tree, even though its proven competitive bound is tighter. 
\end{abstract}

\section{Introduction}

The binary search tree (BST) is a classic data structure for maintaining an ordered set of data with fast queries (referred to as `accesses'). Binary search trees allow following left/right child and parent pointers, along with performing rotations as described further in Section \ref{sec:bst}. BSTs are particularly useful because, if balanced, they allow accesses in time $O(\lg n)$ for a tree on $n$ keys. As a result, many popular BSTs use rotations to ensure the tree is balanced (i.e. height is $O(\lg n)$) during insertions and deletions so that all accesses can be performed in $O(\lg n)$ time. Specific examples include red-black trees \cite{guibas1978dichromatic} and AVL trees \cite{adelson1962algorithm}.

However, these balanced BST algorithms are fundamentally limited in that often cannot do better than $\Theta(\lg n)$ per query. In particular, many real-world access sequences have structure to them (e.g. working-set and sequential accesses) beyond a random access \cite{levy2019foundation}. But, the balanced BSTs, which do not modify their structure during access cannot adjust to take advantage of this structure. As a simple example, consider the access sequence which repeatedly accesses the smallest key in the BST. On a balanced BST such as the red-black tree, each access will take $\Theta(\lg n)$, but if the tree was reorganized to have this element at the root, we could achieve $O(1)$ per access. 

Self-adjusting trees attempt to solve this by performing rotations to reorganize the tree during accesses. The most popular and interesting amongst these is the splay tree \cite{splay}, which achieves optimal performance (within a constant factor) on several access sequences. And, the dynamic optimality conjecture proposes that the splay tree is optimal (up to constant factor) on all BST access sequences. 

However, no stronger competitive bound has been proven on the splay tree besides the trivial $O(\lg n)$-competitiveness (which is obtained since the splay tree has $O(\lg n)$ amortized on any access). The first major progress on competitive binary search trees came from Demaine et al. \cite{tango} with the tango tree, the first proven $O(\lg \lg n)$-competitive BST. However, as discussed later in Section \ref{sec:worst-case}, the tango tree fails to be optimal on several access sequences. Another closely related BST, the multi-splay tree \cite{mst}, is more promising, as it is also $O(\lg \lg n)$-competitive, but achieves optimal performance on many access sequences where the tango tree fails \cite{mst-prop}. 

Finding BSTs with competitive bound is important because it brings us closer to an algorithm that is ideal for any access sequence. In practice, such a tree would display the properties of several specialized data structures by automatically adapting to changing access patterns in an online fashion \cite{levy2019foundation}. 

This work focuses on experimental evaluation of and corresponding theoretical insights related to tango, multi-splay, and splay trees. Our contributions include: developing a family of access sequences which test interesting BST properties; comparing practical performance amongst various BST algorithms; using experimental results to drive theoretical insights about optimality and behavior of these trees.  

To the best knowledge of the authors, no complete experimental implementation of tango trees has been published thus far, and few implementations and experiments have been conducted on the multi-splay tree. In this work, we present the first implementation of the tango tree, and compare the tango tree performance to the multi-splay tree \cite{mst} and the splay tree \cite{splay} to test for optimality on different input sequences. 

After analyzing experimental results, we conduct a deeper analysis of the multi-splay tree's counter-intuitive behavior on the working set access. Motivated by experimental data, we are able to prove new results about the lower order terms in the multi-splay tree running time. We also use a combination of theoretical analysis and experimental results to gain better insight into the drawbacks of the tango tree, and what about its structure leads to slower running time on certain inputs.

\section{Background}

In this section, we present background on BSTs, the interleave bound, and the tango tree. The tango tree is an online binary search tree data structure, and the first to be proven to be $O(\lg \lg n)$-competitive against the optimal offline BST data structure \cite{tango}.

\subsection{BST Computation Model}
\label{sec:bst}

When discussing the optimality of binary search trees, it is necessary to precisely define the types of data structures and operations that are allowed and how the cost of operations is computed. As is standard, we use the model developed by Wilber \cite{wil89}. 

A BST data structure supports accesses on a static set of $n$ keys stored in its nodes. Here, we only consider the universe of keys $\{1, \dots, n\}$ for simplicity. The nodes in a BST are in symmetric order, and each node stores a pointer to its parent, pointers to its left and right children, and any auxiliary information (not including additional pointers). The auxiliary information has size $O (\lg n)$ bits per node, keeping with the convention used by \cite{mst} and others. The input to the BST is an access sequence $X = x_1, x_2, \dots, x_m$ where each access key is chosen from the universe of keys.

The BST access algorithm is allowed to keep a single pointer to a node in the BST. On each access, the pointer is initialized to the root of the tree (with cost 1). Then, the access algorithm can follow any pointer (i.e. parent pointers, left/right child pointers) with unit cost, and can perform rotations with unit cost. In this way, the cost of any access is equal to the sum total number of pointers of followed and rotations performed in the BST. Note that in the BST computation model, modifying the fields of a node (to which the access pointer currently points) is considered free.  

Note that this computation model does not account for insertion or deletions. In this paper, we only consider BSTs which are already constructed, and we define the construction scheme where relevant.

\subsection{Interleave Lower Bound}
\label{sec:ilb}

Both the algorithm and competitive analysis of tango trees are motivated by the interleave lower bound, so we discuss that first. The interleave lower bound serves as a lower bound on the running time of any BST data structure on any given access sequence $X = x_1, x_2, \dots, x_m $. We state the lower bound here to justify the competitive and runtime analysis of tango trees, but do not provide a proof as it is outside the scope of this paper (we refer readers to \cite{tango} for the proof). 

Consider the set of keys $\{ 1, 2, \dots, n \}$ (or re-label the keys in this way), and also consider the static complete binary tree $P$ on these keys. For each node $v$ in $P$, we can define the left region $L(v)$ to be the union of the left subtree of $v$ and $v$ itself, and the right region $R(v)$ the be the right subtree of $v$ (not including $v$). Then, $L(v)$ and $R(v)$ partition the subtree rooted at $v$ such that for every access $x_i$, it satisfies exactly one of the following: 1) $x_i$ is outside the subtree rooted at at $v$; 2) $x_i$ is in $L(v)$; 3) $x_i$ is in $R(v)$. Now, if we consider only those accesses that access elements inside the subtree of $v$, $X^v$, the \textit{interleaving number} of $v$ is the number of times the accesses $X^v$ switch between $L(v)$ and $R(v)$. That is, the interleaving number counts the number of accesses $x^v_i \in X^v$ such that either $x^v_i \in L(v)$ and $x^v_{i+1} \in R(v)$ or $x^v_i \in R(v)$ and $x^v_{i+1} \in L(v)$.

The \textit{interleave bound} for a sequence of accesses $X$ is $\IB(X)$, the sum of the interleaving numbers across all nodes (i.e. all keys) in the complete binary tree $P$. The interleave lower bound is:

\begin{theorem}
\label{thm:ilb}
    The total cost of the optimal offline BST on sequence, $\OPT(X)$, satisfies $\OPT(X) \geq \IB(X)/2 - n$ for any access sequence $X$.
\end{theorem}

\subsection{Tango tree structure}

The tango tree is defined in reference to an augmented version of the complete binary tree $P$ on the keys (as defined in Section \ref{sec:ilb}). We augment $P$ by keeping track of \textit{preferred paths}. Particularly, for every node $v$, its \textit{preferred child} is set based on which of its subtrees was most recently accessed. That is, if a node in $L(v)$ is accessed, $v$'s left child is now its preferred child. And, if a node in $R(v)$ is accessed, $v$'s right child is now its preferred child. The preferred paths of $P$ are the paths connected by following preferred child pointers. An example of the preferred paths decomposition of $P$ on 15 nodes is shown in Fig. \ref{fig:pref}.

\begin{figure}[ht]
    \centering
    \begin{asy}
        size(200);
        import graphs;
        Vertex a1 = Vertex(0, 0, "8");
        
        Vertex b2 = Vertex(-1, -.5, "4");
        Vertex c2 = Vertex(1, -.5, "12");
        
        Vertex b31 = Vertex(-1.5, -1, "2");
        Vertex b32 = Vertex(-0.5, -1, "6");
        
        Vertex c31 = Vertex(.5, -1, "10");
        Vertex c32 = Vertex(1.5, -1, "14");

        Vertex b411 = Vertex(-1.75, -1.5, "1");
        Vertex b412 = Vertex(-1.25, -1.5, "3");
        Vertex b421 = Vertex(-.75, -1.5, "5");
        Vertex b422 = Vertex(-.25, -1.5, "7");

        Vertex c411 = Vertex(.25, -1.5, "9");
        Vertex c412 = Vertex(.75, -1.5, "11");
        Vertex c421 = Vertex(1.25, -1.5, "13");
        Vertex c422 = Vertex(1.75, -1.5, "15");
        
        Vertex.connect(a1,b2);
        Vertex.connect(a1,c2);

        Vertex.connect(b2, b31);
        Vertex.connect(b2, b32);

        Vertex.connect(c2, c31);
        Vertex.connect(c2, c32);

        Vertex.connect(b31, b411);
        Vertex.connect(b31, b412);
        Vertex.connect(b32, b421);
        Vertex.connect(b32, b422);

        Vertex.connect(c31, c411);
        Vertex.connect(c31, c412);
        Vertex.connect(c32, c421);
        Vertex.connect(c32, c422);
        
        Vertex.drawAll(lab=true);
        
        Vertex.highlight(a1, b2);
        Vertex.highlight(b2, b32);
        Vertex.highlight(b32, b421);
        
        Vertex.highlight(c2, c31);
        Vertex.highlight(c31, c411);

        Vertex.highlight(b31, b411);
        Vertex.highlight(c32, c422);
    \end{asy}
    \caption{An example complete binary tree on 15 keys, with one possible collection of preferred paths highlighted. Note that leaves which are not the the preferred child of any node form a preferred path of one node including only themselves.}
    \label{fig:pref}
\end{figure}

The key idea with tango trees is that interleaves are naturally encoded into this preferred path structure, providing a fairly straightforward route to a competitive upper bound. 

\begin{lemma}
\label{lma:interleave}
    Consider the path from the root of $P$ to $x_i$ for some access $x_i$ and let the number of non-preferred edges taken be $k$. The increase in the interleave bound after the access to $x_i$ is $k$. 
\end{lemma}
% \begin{proof}
% Consider some node $v \neq x_i$ along the path from the root to $x_i$ such that the path from the root to $x_i$ does not take $v$'s preferred child. Suppose that $v$'s preferred child is the left child. This means that the most recent access within the subtree of $v$ was within $L(v)$. And, we must have $x_i \in R(v)$, since the path from the root to $x_i$ uses $v$'s right child. Thus, this access increases the interleaving number of $v$ by one. 

% Since there are $k$ non-preferred edges taken, there are $k$ such nodes $v$. And, since each such node has its interleaving number increased by one, the interleave bound must increase by $k$. 
% \end{proof}

So, intuitively, if we can quickly navigate within each preferred path, the transitions between preferred paths are paid for by the increase in the interleave bound (at least from a competitive perspective). The tango tree accomplishes this by storing each preferred path as an auxiliary tree, which are then joined together to form a single BST. The auxiliary trees are implemented using augmented red-black trees, which we describe in depth in Section \ref{sec:auxiliary}.

Specifically, in order to construct the tango tree, we first disconnect the preferred path containing the root of $P$ and transform it into an auxiliary tree. This disconnects $P$ into several subtrees, which are recursively constructed into tango trees. Then, the resultant trees are reattached to the root auxiliary tree as uniquely determined by BST ordering. 

In Fig. \ref{fig:tango1}, we show an example of this construction using the same preferred path decomposition shown in Fig. \ref{fig:pref}. The auxiliary trees are highlighted, although note that the auxiliary trees were constructed arbitrarily here. 

\begin{figure}[ht]
    \centering
    \begin{asy}
        size(250);
        import graphs;

        Vertex v5 = Vertex(0, 0, "5");
        Vertex v6 = Vertex(.5, -.5, "6");
        Vertex v4 = Vertex(-.5, -.5, "4");
        Vertex v8 = Vertex(1, -1, "8");

        Vertex v2 = Vertex(-1, -1, "2");
        Vertex v1 = Vertex(-1.5, -1.5, "1");

        Vertex v3 = Vertex(-.5, -1.5, "3");

        Vertex v7 = Vertex(.5, -1.5, "7");

        Vertex v10 = Vertex(1.5, -1.5, "10");
        Vertex v9 = Vertex(1, -2, "9");
        Vertex v12 = Vertex(2, -2, "12");

        Vertex v11 = Vertex(1.5, -2.5, "11");

        Vertex v14 = Vertex(2.5, -2.5, "14");
        Vertex v15 = Vertex(3, -3, "15");

        Vertex v13 = Vertex(2, -3, "13");
        
        Vertex.connect(v5, v6);
        Vertex.connect(v5, v4);
        Vertex.connect(v6, v8);

        Vertex.connect(v4, v2);
        
        Vertex.connect(v2, v1);
        
        Vertex.connect(v3, v2);

        Vertex.connect(v8, v7);

        Vertex.connect(v8, v10);

        Vertex.connect(v10, v9);
        Vertex.connect(v10, v12);

        Vertex.connect(v12, v11);
        Vertex.connect(v12, v14);

        Vertex.connect(v14, v15);
        Vertex.connect(v14, v13);
        
        Vertex.drawAll(lab=true);

        Vertex.highlight(v5, v6);
        Vertex.highlight(v5, v4);
        Vertex.highlight(v6, v8);

        Vertex.highlight(v2, v1);
        
        Vertex.highlight(v10, v9);
        Vertex.highlight(v10, v12);

        Vertex.highlight(v14, v15);
    \end{asy}
    \caption{An example construction of a tango tree on 15 nodes based on the preferred paths composition shown in Fig. \ref{fig:pref}. Auxiliary trees are highlighted.}
    \label{fig:tango1}
\end{figure}

\subsection{Tango tree access algorithm}
\label{sec:tango-algo}

The tango tree specifies a BST access algorithm. All insertions and deletions are assumed to happen before accesses and are handled according to the red-black tree insertion/deletion algorithm. So, we only describe and analyze the access algorithm here. 

After insertions and deletions are completed, we say that the tree is \textit{locked}. Once the tree locks, we augment every node $u$ with $d_u$, representing the depth of $u$ in the tree when locking the tree. Let $\AT(u)$ be the auxiliary tree that $u$ belongs to. Every node $u$ initially will be its own auxiliary tree, denoted by $\AT(u) = u$.

The tango tree access algorithm reflects a typical BST walk. When the algorithm attempts to move from node $u$ to another auxiliary tree rooted at $v$, it starts by cutting $\AT(u)$ at depth $\min_{p\in \AT(v)} d_{p}-1=d'$, meaning that all nodes of depth larger than $d'$ are separated and made into their own auxiliary tree. Then, it joins $\AT(u)$ with $\AT(v)$ into one single auxiliary tree. This process causes the preferred child across the path to swap into $v$. The algorithm then continues from $v$. For simplicity, our implementation does these cuts and joins lazily after we reach the end of the access.

One could optimize the constant factor of the tango tree by choosing an order of these cuts and joins such that when joining two auxiliary trees they will be in similar size; as the complexity of joining two auxiliary trees is linear in the difference of the depths of the two trees. 

In order to implement these operations we augment every node $u$ with $a_u$ and $b_u$, representing the minimum and maximum depth of any node in $\AT(u)$ which is in the subtree of $u$, respectively. We need to maintain these values when rotating, cutting, and joining auxiliary trees.

\subsection{Tango tree auxiliary trees}
\label{sec:auxiliary}

The cut and join operations can be implemented using split and concatenate operations as explained in \cite{cutjoin}. To split an auxiliary tree at depth $d$, we find $l'$ and $r'$ such that all nodes with values in the range $(l',r')$ have depth larger than $d$, and all nodes outside that range have depth smaller than or equal to $d$. This is always true as each auxiliary tree is composed of nodes of unique and continuous depth values, so a continuous subsection of key values in the auxiliary tree corresponds to a continuous range of depth values. Then we split the auxiliary tree at $l'$ and then at $r'$. We then mark the left child of $r'$ to be a new auxiliary tree (which corresponds to nodes of depth larger than $d$). Finally we concatenate $r'$ and then $l'$. Joining is simply the same process but backwards.

\subsection{Tango tree running time and competitive bound}

\begin{lemma}
\label{lma:accesstime}
    Let $k$ be the number of preferred child pointers updated by the access to $x_i$ on the tango tree. Then, the running time of the access is $O ((k+1)(1 + \lg \lg n))$. 
\end{lemma}

\begin{proof}
    First, note that the access consists of two parts --- first, traversing the tango tree from the root to $x_i$; second, rearranging the tango tree by updating preferred child pointers and rebalancing auxiliary trees. 

    First, let us consider the cost of the search from the root to $x_i$. Since each auxiliary tree represents a preferred path in the complete binary search tree (which is perfectly balanced), we can see that each auxiliary tree in the tango tree has $O (\lg n)$ elements. And, since auxiliary tree is implemented using a red-black tree which is balanced, the cost to search through any given auxiliary tree is $O(\lg \lg n)$. Since we update $k$ preferred child pointers, we must have touched at most $k+1$ distinct auxiliary trees. So, the total cost of searching is $O((k+1) (1 + \lg \lg n))$.

    Rearranging the tango tree involves updating $k$ preferred child pointers and rebalancing up to $k+1$ auxiliary trees. As described in \ref{sec:tango-algo}, updating preferred child pointers just involves a constant number of cuts, joins, and rebalancings, each of which takes $O (\lg \lg n)$ on the auxiliary red-black trees (since each has $O (\lg n)$ elements) \cite{cutjoin}. So, the cost of rearranging the tango tree during access is also $O((k+1)(1 + \lg \lg n))$.
\end{proof}

\begin{theorem}
    The total running time of the tango tree algorithm on an access sequence $X$ of size $m$ is $O((\OPT(X) + n)(1 + \lg \lg n))$.
\end{theorem}
\begin{proof}
    Let us suppose that all preferred child pointer updates in access $x_i$ are from left-to-right or right-to-left (as opposed to being set to either left or right from previously having no preferred child). Then, we can see that for access $x_i$, the number of updated preferred child pointers is precisely equal to the number of additional interleaves due to access $x_i$, as per Lemma \ref{lma:interleave}. We will let $\IB_i(X)$ be the number of additional interleaves due to access $x_i$. Note that $\sum_i \IB_i(X) = \IB(X)$.

    Since each node starts with no preferred child, there are at most $n$ preferred child updates that are not left-to-right or right-to-left, each with cost $O(1 + \lg \lg n)$. So, from Lemma \ref{lma:accesstime}, we see that the total access time for access sequence $X$ is 
    
    \begin{align*}
        &O\left( \left (\sum_i (\IB_i(X) + 1) + n \right)(1 + \lg \lg n)\right)\\
        = \, &O((\IB(X) + n + m)(1 + \lg \lg n)).
    \end{align*}
    
    From Theorem \ref{thm:ilb}, we know that $\OPT(X) \geq \IB(X) / 2 - n$, and we can see that $\OPT(X) \geq m$ is a a trivial lower bound, so we get that the total tango access time is $O((\OPT(X) + n)(1 + \lg \lg n))$.   
\end{proof}

And, we can easily see that the following corollary holds, giving $O(\lg \lg n)$-competitiveness as desired:
\begin{corollary}
\label{cor:competitive}
    When the size of the access sequence $|X| = m$ satisfies $m = \Omega(n)$, the running time of the tango tree algorithm is $O(\OPT(X)(1 + \lg \lg n))$. 
\end{corollary}

\subsection{Worst-case Access Sequences for Tango Trees}
\label{sec:worst-case}

While the $O(\lg \lg n)$-competitiveness of the tango tree is an important result, the worst case competitive running time does exist on certain access sequences. Particularly, let us consider the sequential access, in which we have the access $X = 1, 2, \dots, n $. It is known that several BSTs achieve $O(n)$ total running time on the sequential access, including splay trees \cite{seq-splay}, which is optimal. However, we can show that the tango tree has $\Theta(n \lg \lg n)$ running time on the sequential access.  While several papers refer to this fact (e.g. \cite{mst}), we have yet to see a formal proof or explanation of the above. The full proof is worked through in Appendix \ref{sec:tango-worst-case}, but we present some intuition here. 

The key limiting factor in the tango tree is the fact that navigating to any leaf of a red-black tree of size $k$ takes $\Theta(\lg k)$ time. This means that each traversal between auxiliary trees takes $\Theta(\lg |\AT|)$ time, which introduces an overhead of $\Theta(\lg \lg n)$ from the optimal BST access time on many accesses (including the sequential access). That is, the use of red-black trees for the auxiliary trees severely limits the ability of the tango tree to achieve optimality on many access sequences.

On the other hand, multi-splay trees, which use splay trees for their auxiliary trees instead of red-black trees, are able to achieve $O(n)$ running time on the sequential access \cite{mst}. Note that multi-splay trees achieve this running time while also achieving $O(\lg \lg n)$-competitiveness. 

Motivated by this distinction, in the next sections we focus on comparing the performance of the tango tree with the multi-splay tree and the splay tree. 

\section{Experimental Methodology}

We present, what is, to the best knowledge of the authors, the first experimental implementation of tango trees. We also implement multi-splay trees and splay trees to compare with the tango tree. We chose to study the multi-splay tree and splay tree alongside the tango tree because the splay tree is widely conjectured to be dynamically optimal, and because the multi-splay tree is another $O(\lg \lg n)$-competitive BST which is also conjectured to be dynamically optimal. We don't cover the details of the splay tree and multi-splay tree here for brevity, but instead refer readers to the original papers \cite{mst} and \cite{splay}. We will note that the multi-splay tree is extremely similar in design to the tango tree except that it uses splay trees as its auxiliary trees rather than red-black trees.  

As noted above, we present a complete implementation of the tango tree algorithm for the first time in this paper. However, previous work has been done on implementing the multi-splay and splay tree algorithms, which we replicate here for comparison purposes. Particularly noteworthy is the multi-splay tree implementation by Altamarino and Rule \cite{ms-impl}, which we use as a reference for our own implementation. Altamarino and Rule also give an incomplete implementation of tango trees, which served as a good reference material for us to understand some of the challenges in implementing tango trees. 

We implemented all three data structures using C++ because of its fast performance, simple use of pointers, and familiarity to the authors.
In order to account for differences in platforms, instead of measuring the physical running time of each of the algorithms, we instead count the number of unit-cost BST operations executed by each algorithm. Recall from Section \ref{sec:bst} that a unit-cost BST operation involves either following node pointers (left/right child, parent pointers) or executing a rotation on a node. As a result, our implementation effectively measure the sum total number of (not necessarily distinct) nodes touched and number of rotations executed during an access sequence. In particular, we do not measure the cost of modifying the internal data stored by each node. We chose to measure unit cost BST operations instead of running time to provide verifiable platform independent results, and because our focus is more on the large trends affecting algorithm performance rather than on system-level optimizations that may provide constant factor speed-ups. 

We implement a number of different test sequences with the intent of measuring specific properties of each of the BST algorithms. Our goal is to create a family of access sequences which can be used to provide insights into whether a given BST algorithm is performing optimally. We list each of these tested properties below and provide some explanation into the structure of the access sequences used as well.

Our full implementations and experiments are available open-source on GitHub \cite{github}.

\subsection{Sequential Access Property}

As described in Section \ref{sec:worst-case}, the sequential access is simply the access sequence $X = 1, 2, \dots, n$. The sequential access property is the property that a BST can complete the sequential access in $O(n)$ total time (or equivalently, in $O(1)$ amortized cost per operation). We test this property by making several passes\footnote{The number of passes is a constant, in this case, we use 25 passes.} over the sequential access and measuring the cost per access as the total cost of the access sequence divided by the total number of accesses. We repeat this test as we vary $n$, so this normalization step is important to allowing us to compare amortized cost per operation. Otherwise, there is a dominating factor of $n$ in the sequence access cost which just represents the fact that there are $\Theta(n)$ total accesses made. We use this normalization to amortized cost idea in other tests as well.

\subsection{Random Access}
In this test, we do not test a particular property, but rather test the BST algorithms on a uniformly random access sequence. That is, we make $m$ accesses to a BST on $n$ keys where each access is chosen uniformly at random from the universe of keys $\{ 1, \dots, n\}$. We use $m = 25 \cdot n$, and for each BST, we measure the average cost per access as $n$ varies. 

\subsection{Working Set Property}

The working set property states that if $t$ distinct keys are accessed since the last access to some key $x$, then key $x$ can be accessed in $O(\lg t)$. Satisfying this property implies that if there is a working set of keys $S$ such that all accesses in a sequence satisfy $x_i \in S$, the BST can complete these accesses in $O(\lg |S|)$ time per operation. 

So, in order to test the working set property we randomly partition the key universe into $n/k$ sets of $k$ keys each. Then, for each set of $k$ keys, we make several accesses to each element in the working set.\footnote{For a working set of size $k$, we make $100 \cdot k$ accesses, or 100 accesses per element.} We measure the average cost per access for each data structure as we vary both $n$ and $k$. If a BST satisfies the working set property, we expect its cost per operation to be $O(\lg k)$ (i.e. essentially no dependence on $n$ even as it varies). 

In our experiments, the results for multi-splay trees were hard to interpret. So, we also ran a specific experiments fixing $k=2$ and $k=4$, with much finer resolution on values of $n$ than the initial experiment. The motivation behind this experiment and the results are explained in Section \ref{sec:res-wset}.

\subsection{Unified Property}

The unified property states that if $t_{ij}$ distinct keys are accessed between accessing $x_i$ and $x_j$, then any $x_j$ can be accessed in $O(\lg( \min_i t_{ij} + |x_i - x_j| + 2))$. Note that the unified property implies the working set property.

To check if a data structure has the unified property, we divide the total number of elements into chunks of size $2k$, where $k$ is a parameter we choose. We then group these chunks into sets of $k$ chunks each, spread uniformly. We then iterate over each set of chunks and access specific elements in a specific order. During iteration $2i$, we access the $i$th element of each chunk in increasing order, and during iteration $2i+1$, we access the $(k+i)$th element of each chunk in increasing order. This guarantees that when we access an element $x_j$, there exists another element $x_i$ such that $t_{ij}=k$ (where $t_{ij}$ is as defined above) and $|x_i - x_j| < k$ (ignoring the first iteration), which should take $O(\lg k)$ time per access if the unified property holds. We verified our code is correct by calculating the average unified bound for the access sequence. Note this access sequence is only correct up to $O(\sqrt{N})$ as above that point, we can't divide them into enough groups.

However, note that a data structure that efficiently processes this specific access sequence may not necessarily have the unified property since this sequence has a specific structure and is not randomly chosen from all access sequences where the unified property holds for some $k$.

\section{Results and Discussion}

The full data collected for each test is available on our project GitHub repository \cite{github}. Here we summarize the results in graphical format and discuss some of the implications. 

\subsection{Sequential Access}
The average access costs for sequential access on $n$ keys is shown in Fig. \ref{fig:seq}. Both the splay tree \cite{seq-splay} and the multi-splay tree \cite{mst} are known to take $O(1)$ amortized time per access on the sequential access, while we proved that the tango tree takes $\Theta(\lg \lg n)$ amortized time per access. As expected, then, the access cost for the tango tree is higher than both the multi-splay tree and splay tree on the sequential access. And particularly, the average cost for the tango tree grows roughly linearly with $\lg \lg n$, while the splay and multi-splay costs are roughly constant. 

In order to demonstrate further, we also plot the ratio of the tango tree cost to the multi-splay and splay tree costs in Fig. \ref{fig:seq-ratio}. Here, we see that the data does indeed support a $\Theta (\lg \lg n)$ ratio between the tango tree cost and the other algorithms since the cost ratio is roughly linear in $\lg \lg n$. We also see that, although the multi-splay and splay trees both have constant cost per access, the splay tree is more efficient in practice, requiring roughly half as many operations per access. 

This test is mostly used as experimental verification of our analytical result from Section \ref{sec:worst-case}, and also as a sanity check of the correctness of our implementations.

\begin{figure}[ht]
    \centering
    \begin{subfigure}[t]{0.49\textwidth}
        \centering
        \includegraphics[height=2in]{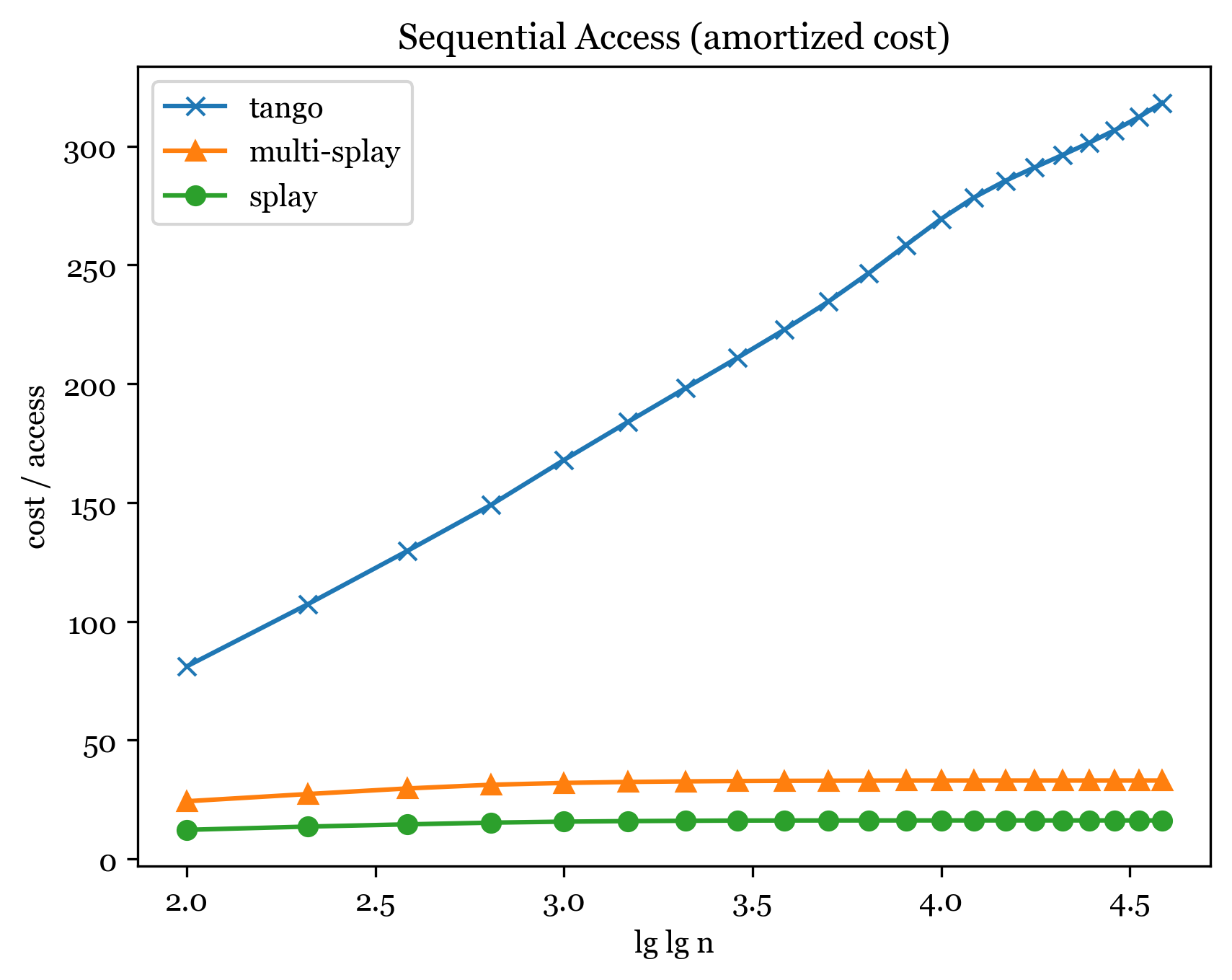}
        \caption{Average sequential access costs for each of tango tree, multi-splay tree, and splay tree.}
        \label{fig:seq}
    \end{subfigure}
    \hfill
    \begin{subfigure}[t]{0.49\textwidth}
        \centering
        \includegraphics[height=2in]{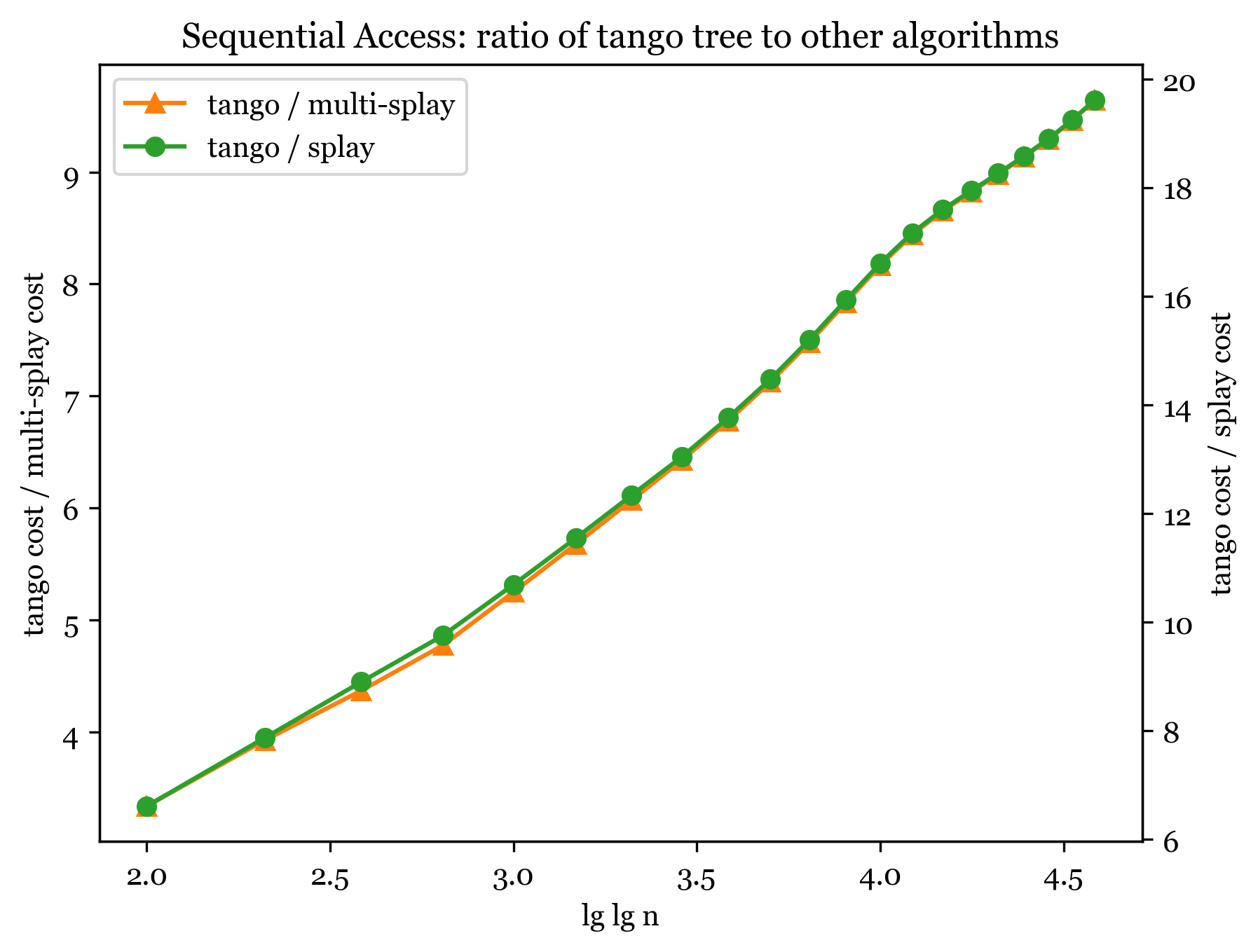}
        \caption{Ratio of the tango tree access costs with the multi-splay tree and splay tree costs. 
        % Note the distinct vertical axes.
        }
        \label{fig:seq-ratio}
    \end{subfigure}
    \caption{Results on the sequential access with $n$ keys. 
    % Note the double log scale on the horizontal axis in both plots.
    }
\end{figure}

\subsection{Random Access}

It is known that on a uniformly random access sequence, the splay tree takes time $O(\lg n)$ per access \cite{splay}. Additionally, Wang et al. claim that the tango tree takes average time $\Theta(\lg n \lg \lg n)$ per access on a random access pattern \cite{mst}. So, this test allows us to experimentally verify this claim, as well as understand where the running time of the multi-splay tree lies respective to the other two algorithms.

% TODO: clarify splay tree is O(lg n) amortized despite flat plot
In Fig. \ref{fig:random}, we plot the average cost per access for each of the three BST algorithms. As with other access sequences, the splay tree greatly outperforms the multi-splay tree, which in turn greatly outperforms the tango tree. In order to understand the relationship between the running times better, we also plot the ratio of access costs between tango tree and the other algorithms in Fig. \ref{fig:random-ratio}. We can clearly see that both the (tango / multi-splay) and (tango / splay) ratios are linear in $\lg \lg n$, which would confirm that the tango tree has average cost $\Theta(\lg n \lg \lg n)$ per access. The data also seems to suggest that the multi-splay tree has average cost $\Theta(\lg n)$ per access on the random access, giving experimental support to the conjecture that the multi-splay tree is dynamically optimal \cite{mst}.

\begin{figure}[ht]
    \centering
    \begin{subfigure}[t]{0.49\textwidth}
        \centering
        \includegraphics[height=1.95in]{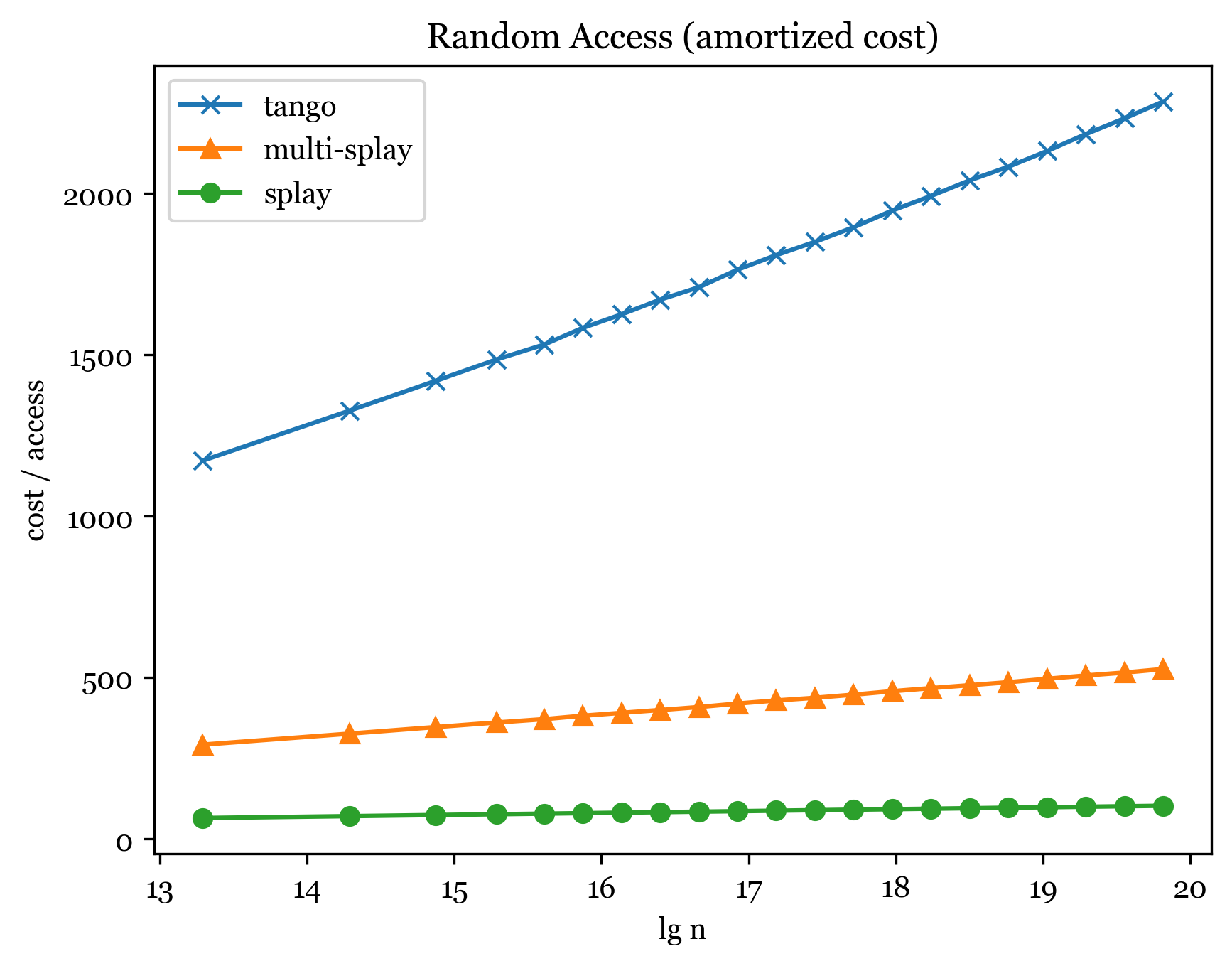}
        \caption{Random access costs for each of tango tree, multi-splay tree, and splay tree. 
        % Note the log scale on the horizontal axis.
        }
        \label{fig:random}
    \end{subfigure}
    \hfill
    \begin{subfigure}[t]{0.49\textwidth}
        \centering
        \includegraphics[height=1.95in]{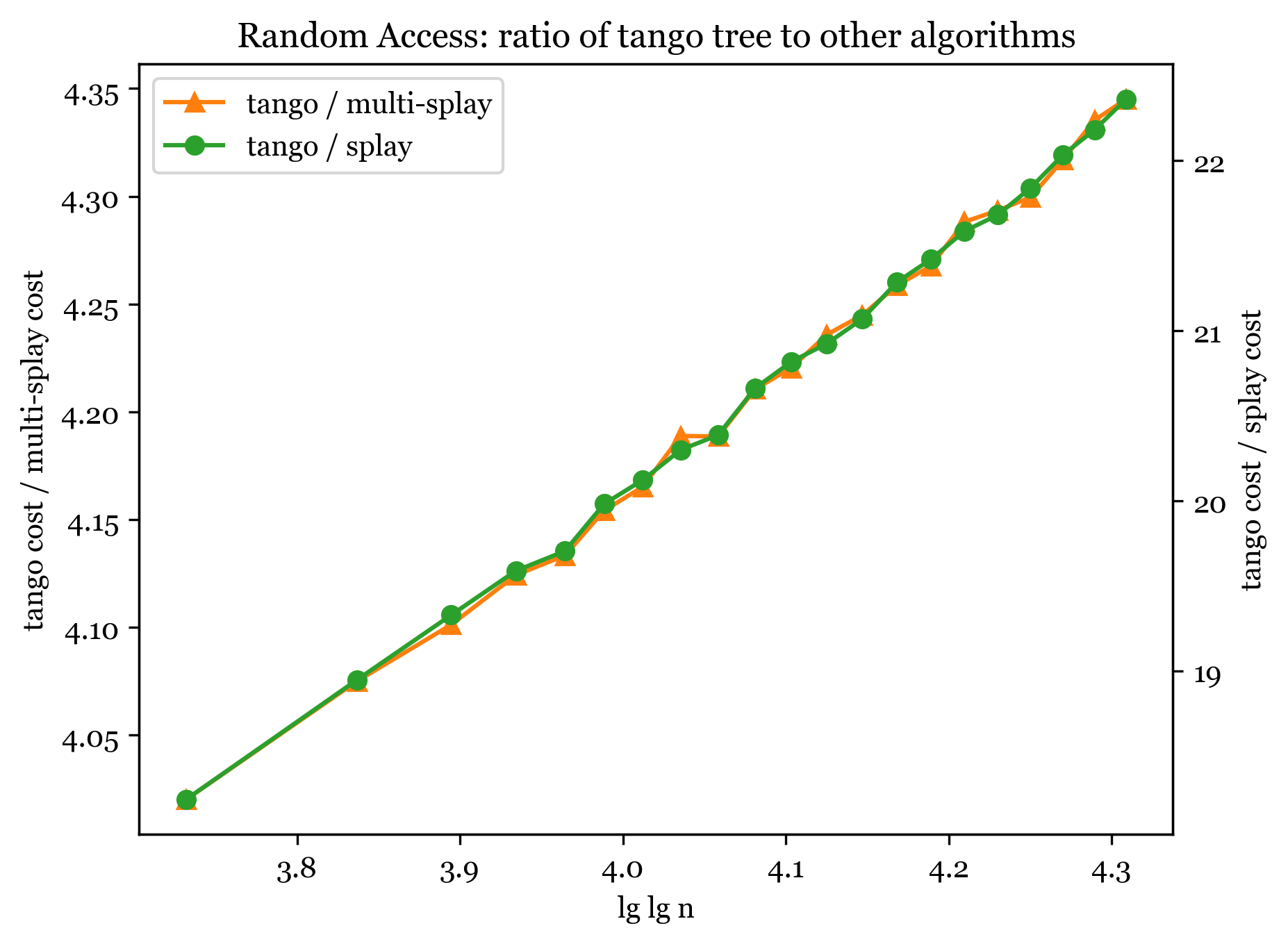}
        \caption{Ratio of the tango tree access costs with the multi-splay tree and splay tree costs. 
        % Note the double log horizontal scale and the distinct vertical axes.
        }
        \label{fig:random-ratio}
    \end{subfigure}
    \caption{Results on the random access with $n$ keys.}
\end{figure}

\subsection{Working Set Property}
\label{sec:res-wset}

Both the splay tree \cite{splay} and the multi-splay tree \cite{mst-prop} have been proven to satisfy the working set property. To our best knowledge, there is no tight bound on the running time of tango tree on a working-set, but from the competitiveness we know the average cost per access should be $O(\lg k \lg \lg n)$. So, we expect to see the splay and multi-splay tree to have access costs dependent only on $k$, while the tango tree may have up to an extra multiplicative $O(\lg \lg n)$ factor. 

First, we examine the results on the splay tree, shown in Fig. \ref{fig:working_set_splay} since these are most straightforward. We can see that when we set the size of the working set $k$ constant, the average cost per access for the splay tree does not change as a function of $n$. And, when we instead control for the number of keys $n$, the average cost per access grows linearly with $\lg k$. Both of these are demonstrated precisely in Fig. \ref{fig:working_set_splay}, providing experimental verification of the fact that splay trees satisfy the the working set property.

\begin{figure}[ht]
    \centering
    \includegraphics[width=.49\linewidth]{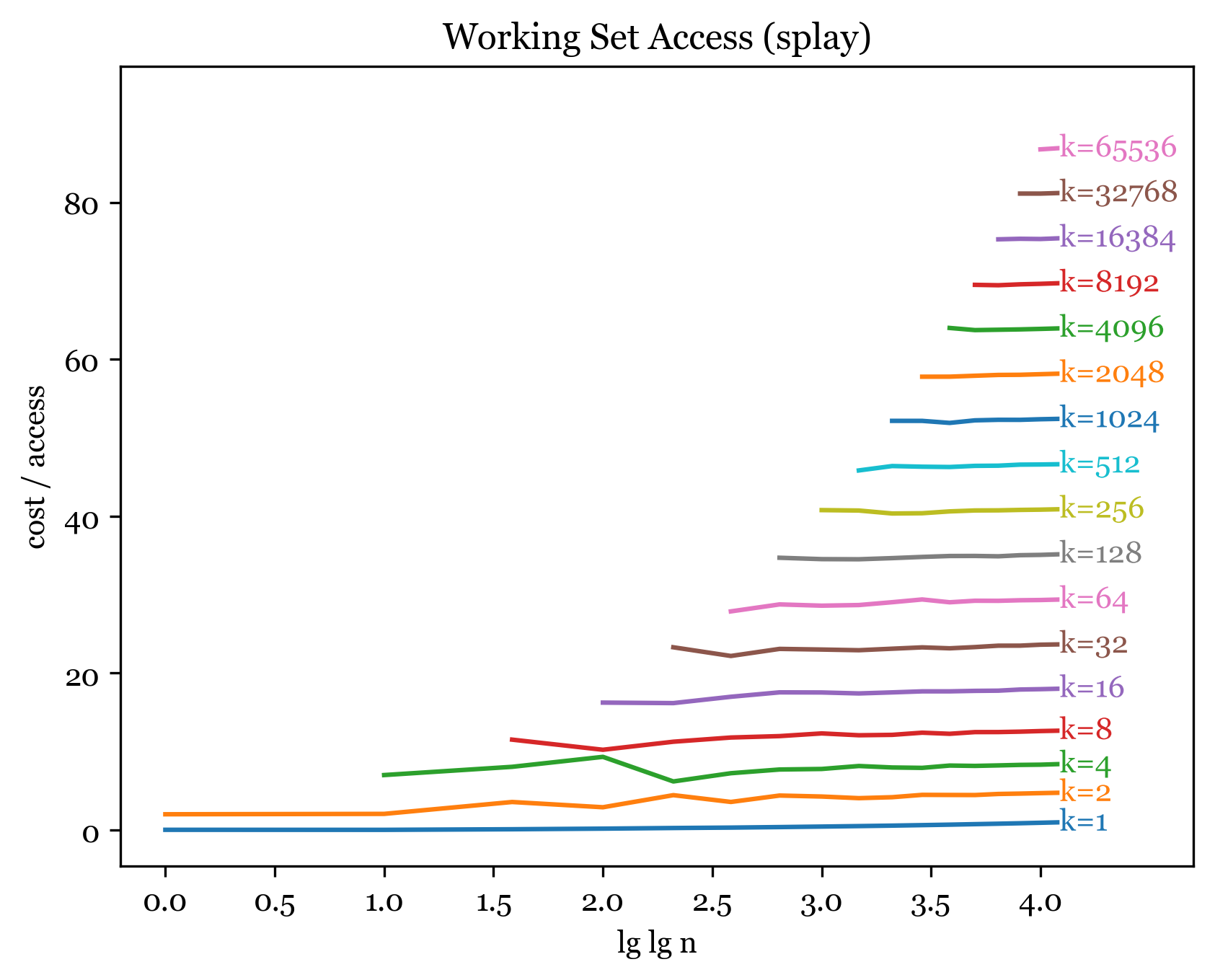}
    \includegraphics[width=.49\linewidth]{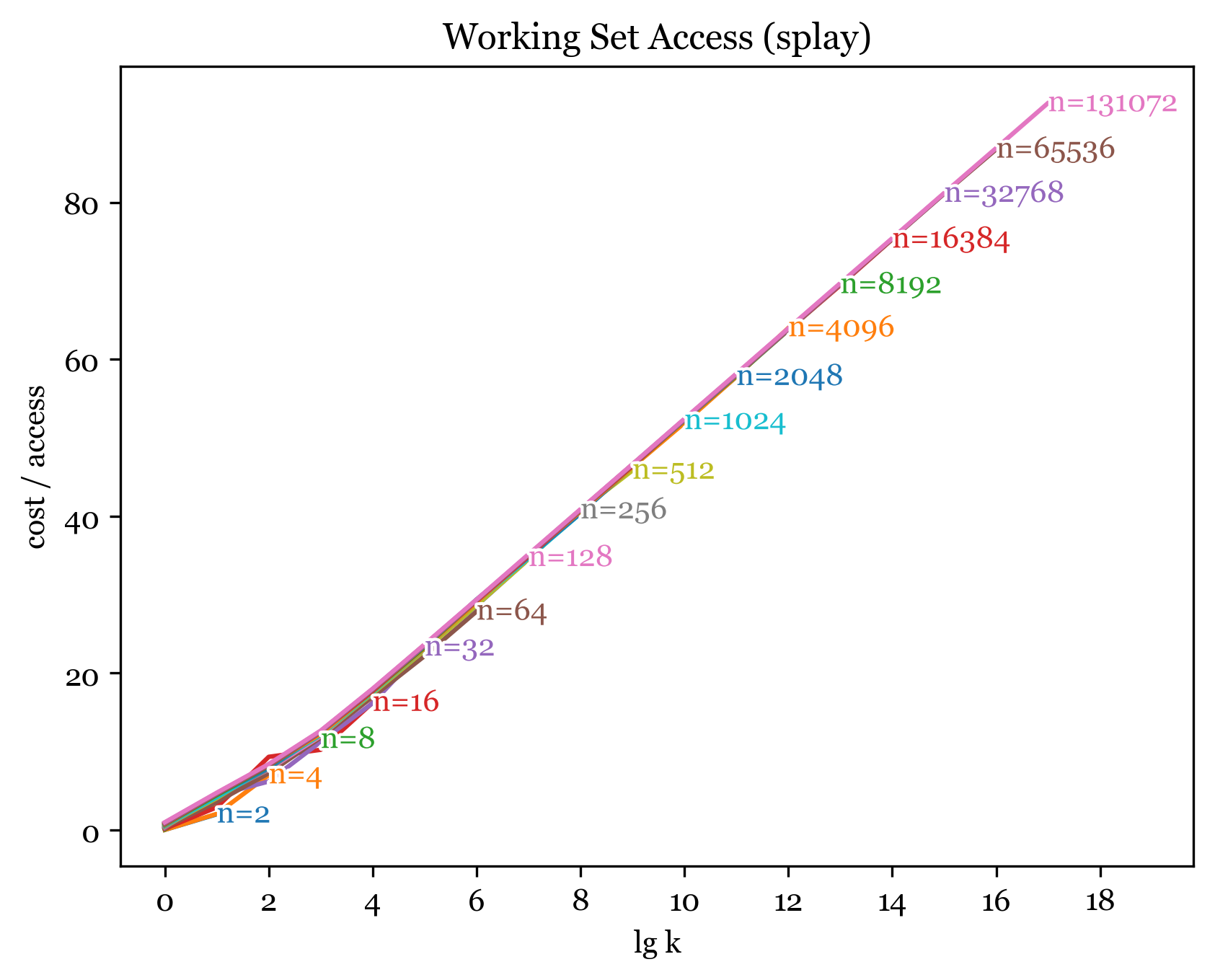}
    \caption{Results of running a working set access sequence on the splay tree, as $n$, the number of keys, and $k$, the size of the working set vary.}
    \label{fig:working_set_splay}
\end{figure}

Now, we turn to tango trees which also lend themselves to easy analysis. As shown, in Fig. \ref{fig:working_set_tango}, when $n$, the number of keys is held constant, the average cost per access grows linearly with $\lg k$, as we might expect. And, when the working set size $k$, is held constant, the average cost per access grows linearly with $\lg \lg n$. So, this indicates that the tango tree likely has running time $\Theta(\lg k \lg \lg n)$ on the working set access.

We can understand why this might be the case by considering the structure of the tango tree, and using insights from our analysis of the sequential access in Section \ref{sec:worst-case}. Because the auxiliary trees used in tango trees are balanced red-black trees, any interleave requires time at least $\Omega(\lg \lg n)$ since we have to navigate to a leaf of depth $\Theta(\lg \lg n)$ in the preferred path of the root of $P$ (since this auxiliary tree has size $\Theta(\lg n)$). So, any time we access a key outside the preferred path, our minimum cost for access is $\Omega(\lg \lg n)$. If we now consider any randomly selected working set of size $k$, and any uniformly random access on this working set, it is intuitively clear that it is very unlikely for these accesses to all land on a single root-to-leaf path of $P$. In particular, we expect $\Theta(k)$ of the selected keys to lie off the preferred path at any given point in time, so that we expect to make $\Theta(m)$ interleaves just on the preferred path of the root for any sequence of $m$ accesses. And, if we accounted for all the other interleaves necessary, it becomes clear why the multiplicative factor of $\Theta(\lg \lg n)$ is present. We don't provide formal proofs, but rather use these explanations to provide intuitive insight for why the tango tree is slow on most inputs (including the working set).

\begin{figure}[ht]
    \centering
    \includegraphics[width=.49\linewidth]{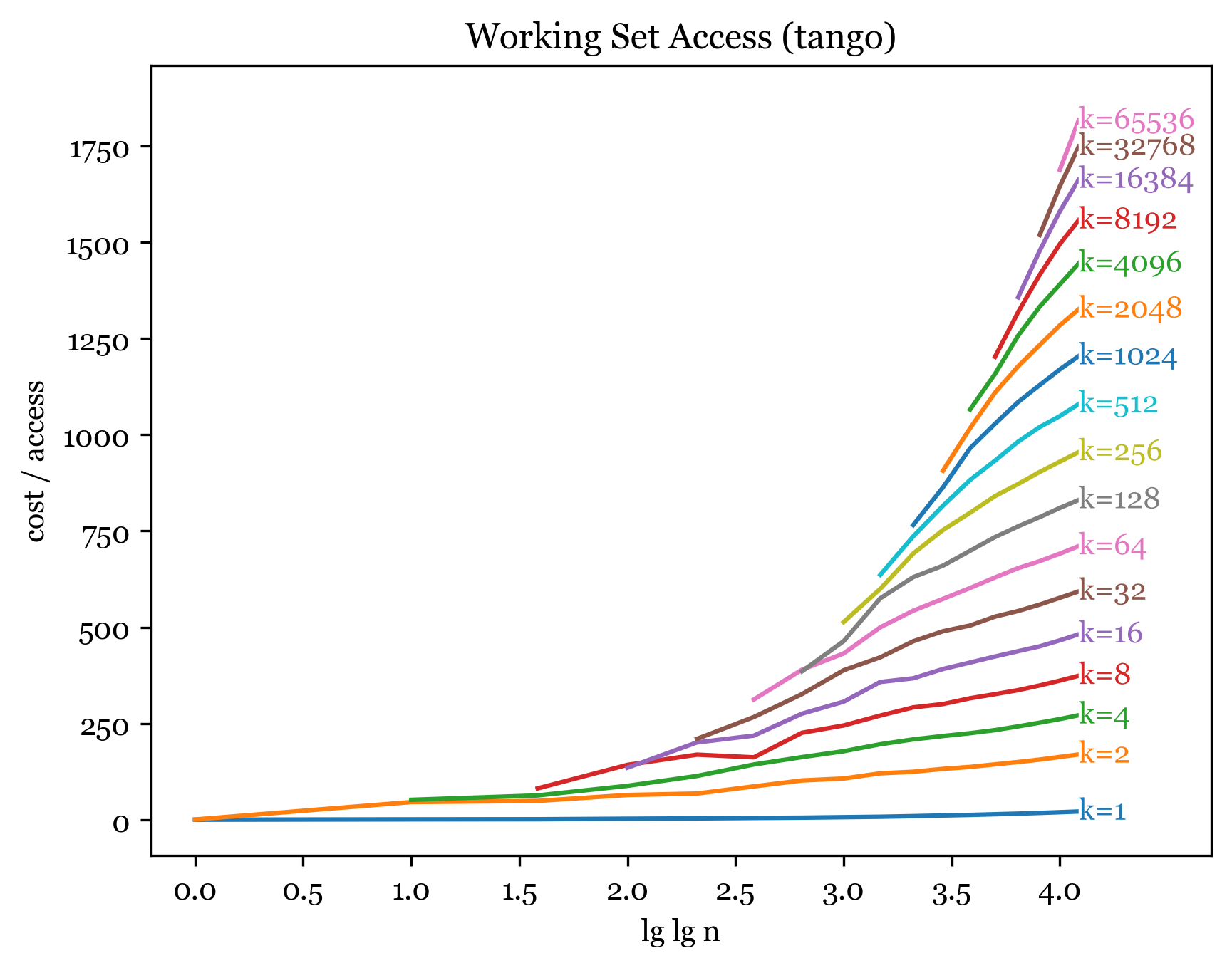}
    \includegraphics[width=.49\linewidth]{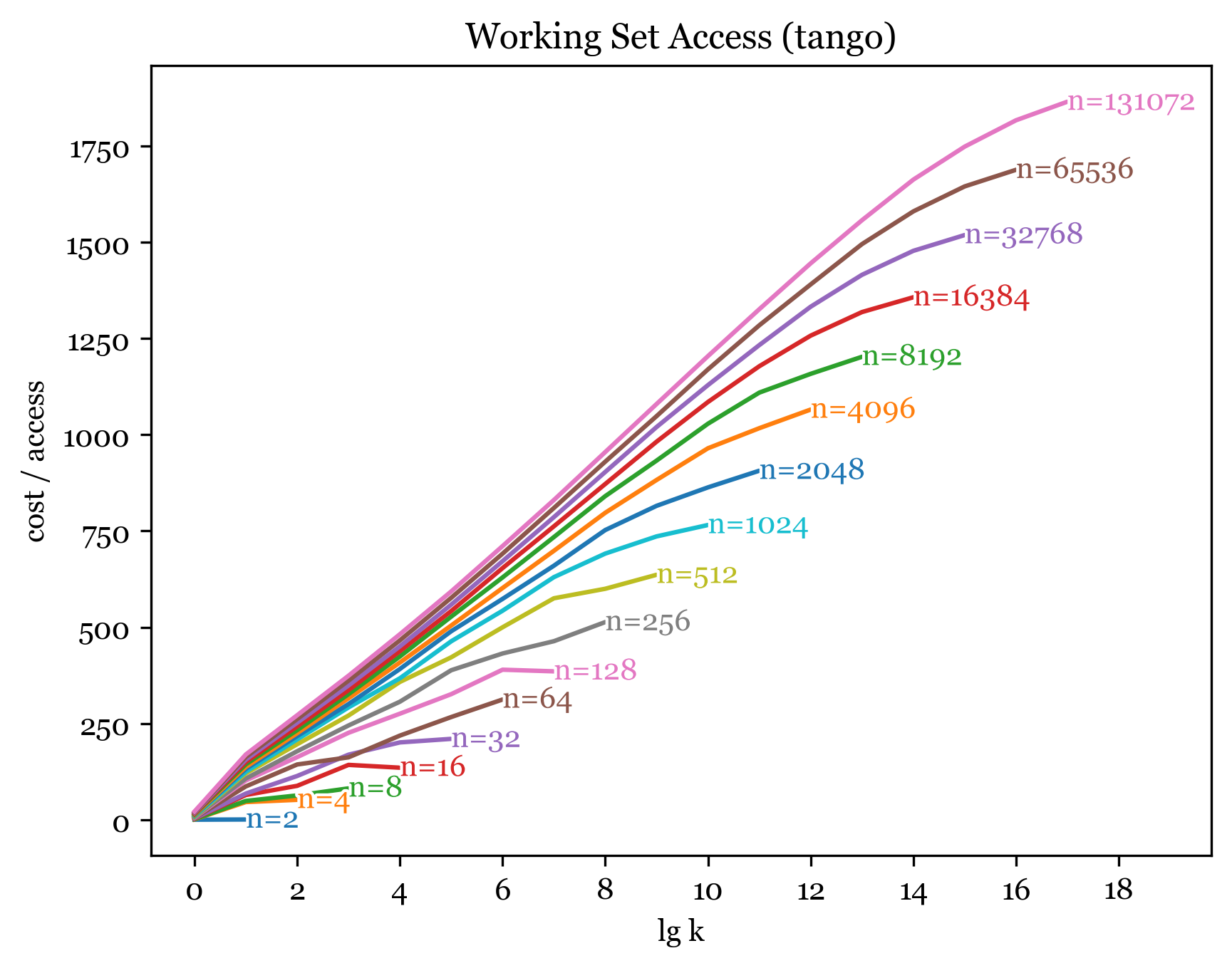}
    \caption{Results of running a working set access sequence on the tango tree, as $n$, the number of keys, and $k$, the size of the working set vary.}
    \label{fig:working_set_tango}
\end{figure}

Now, we turn our attention to the multi-splay tree. In Fig. \ref{fig:working_set_multisplay}, we see that as $n$ is held constant, the average cost per access is roughly linear with $\lg k$, as expected. However, when $k$ is held constant, the average cost per access seems to have some positive correlation with $\lg \lg n$. This is in stark contrast to the results for the splay tree in Fig. \ref{fig:working_set_splay}, which clearly showed no correlation of running time with $\lg \lg n$. Furthermore, this is also in conflict with the theoretical result from Derryberry et al. \cite{mst-prop}, who proved that multi-splay trees do in fact satisfy the working set property. 

% Slopes of lines in right figure should be proportional to the p(1 - lgn / n) factor as estimated in next section
\begin{figure}[ht]
    \centering
    \includegraphics[width=.49\linewidth]{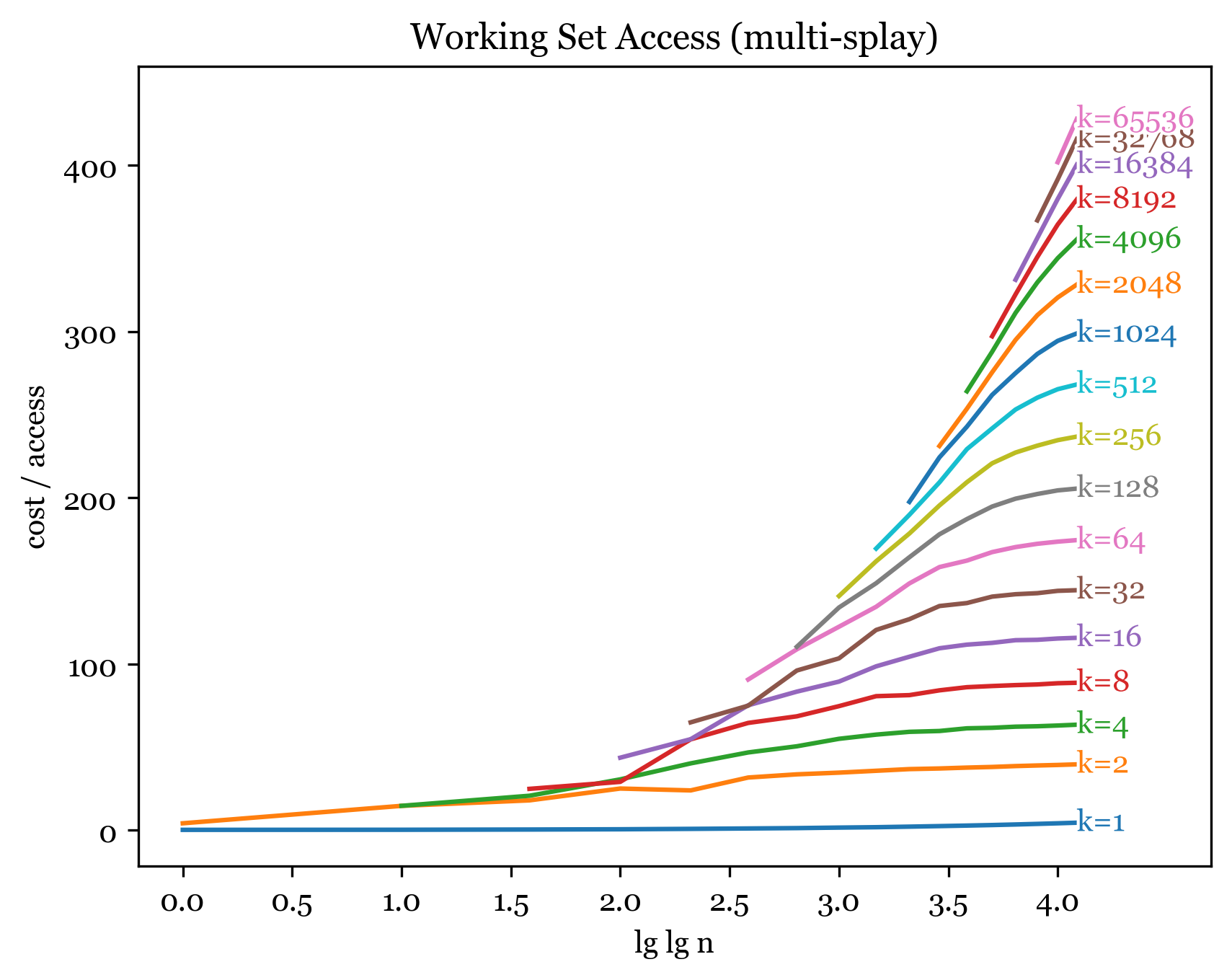}
    \includegraphics[width=.49\linewidth]{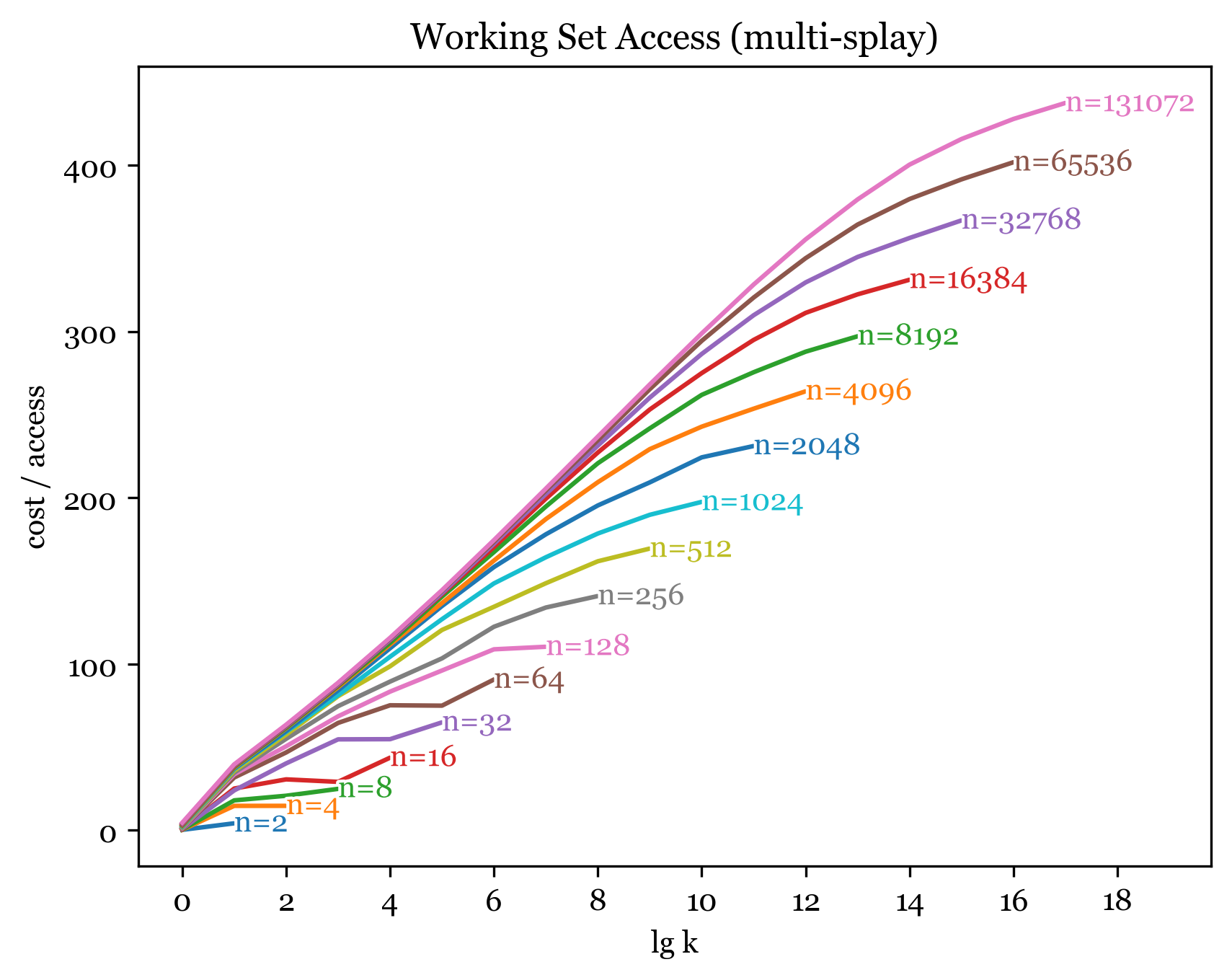}
    \caption{Results of running a working set access sequence on the multi-splay tree, as $n$, the number of keys, and $k$, the size of the working set vary.}
    \label{fig:working_set_multisplay}
\end{figure}

\subsubsection{Further analysis of multi-splay working set access}
\label{sec:ms-wset}

In order to understand this further, we chose to fix $k=2$ and $k=4$ and run accesses for several more values of $n$, spaced evenly in the double log space. The results of this secondary experiment are shown in Fig. \ref{fig:working_set_multisplay_k}. Here, we see that initially there is a positive correlation with the average access cost and $\lg \lg n$, but it levels out beyond $\lg \lg n = 3.5$ (corresponding to $n \approx 2500$). This indicates that the asymptotic running time of multi-splay tree is indeed $O(\lg k)$ per access (i.e. no asymptotic dependence on $n$), in agreement with the theoretical result. 

\begin{figure}[ht]
    \centering
    \includegraphics[width=0.49\linewidth]{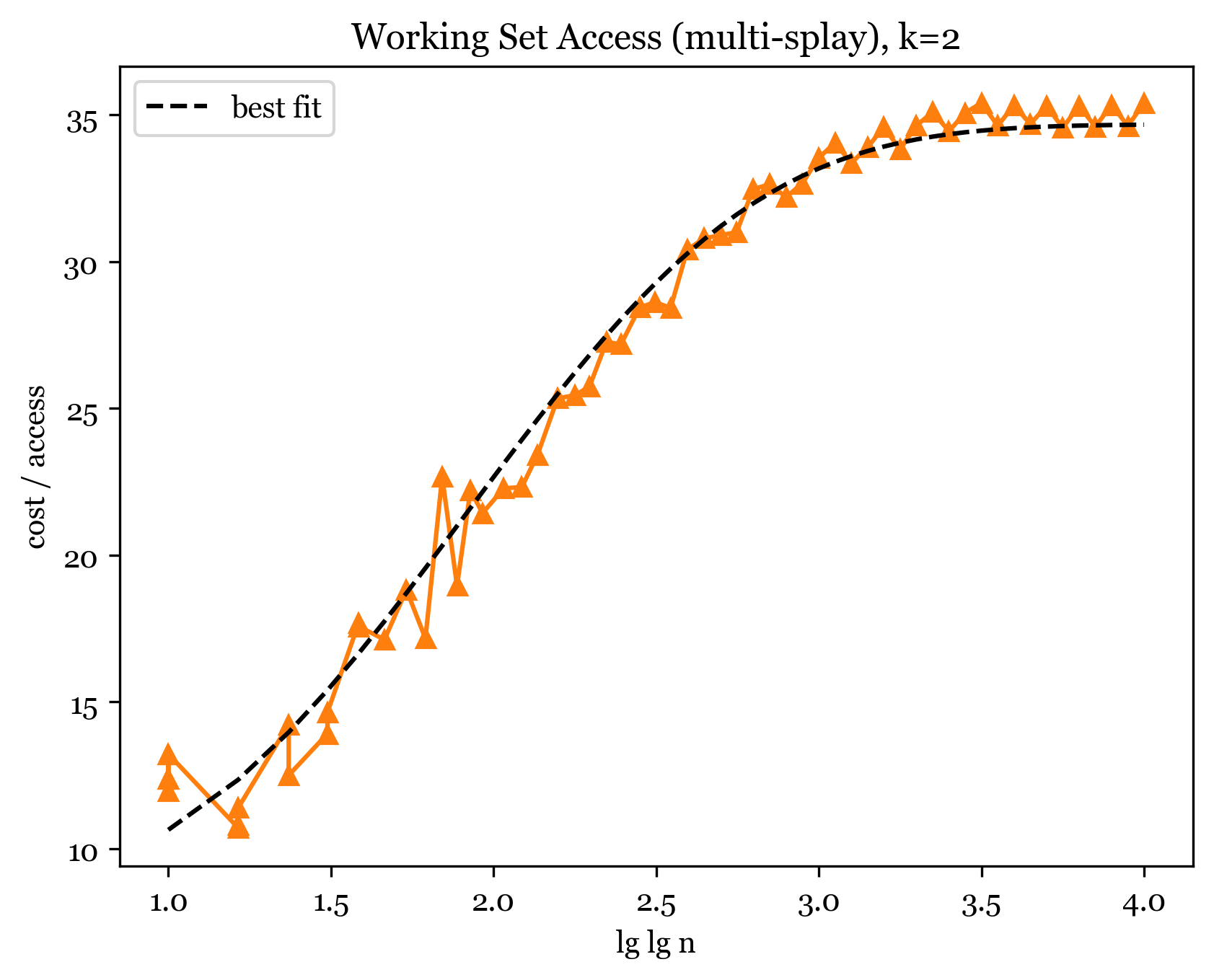}
    \includegraphics[width=0.49\linewidth]{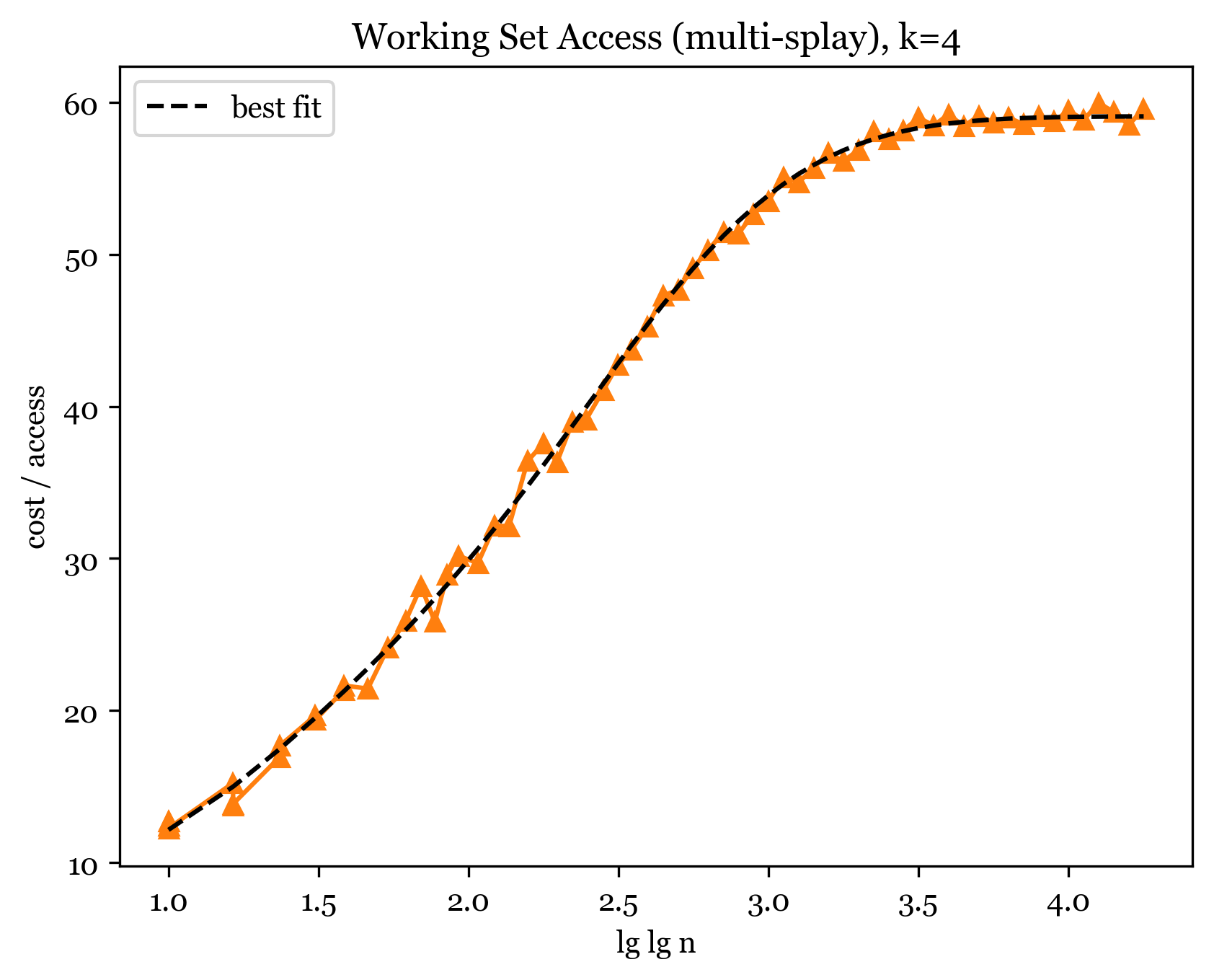}
    \caption{Results of running a working set access sequence on the multi-splay tree, as $n$ varies and $k=2$ or $k=4$ is fixed. Best fit curves are also included; see text and particularly Equations \eqref{eq:fitk2} and \eqref{eq:fitk4} for details.}
    \label{fig:working_set_multisplay_k}
\end{figure}

However, the data does suggest that there are some lower order terms hidden in the running time of the multi-splay tree. We attempt to provide a more precise explanation for the behavior exhibited in Fig. \ref{fig:working_set_multisplay_k}. Notice that since we select our $k$ keys for the working set randomly, some keys may lie on the same preferred path in $P$ with some probability, and thus in the same auxiliary tree in the multi-splay tree. The probability of keys lying in the same auxiliary trees could explain varying access costs. For $k=2$, we can analyze this effect precisely. 

\begin{lemma}
\label{lma:multisplay-k2}
Consider a randomly selected working set $S = \{ x, y \}$ of size 2 and some access sequence on this working set. Let $\rho(n)$ be the probability that $x$ and $y$ lie on the same root-leaf path in $P$. Then, the expected amortized cost per access on the multi-splay tree is $\rho(n)A + (1- \rho(n))B$ for constants $B > A$.
\end{lemma}
\begin{proof}
    If $x$ and $y$ lie on the same root-leaf path of $P$, each access will require no interleaves, since $x$ and $y$ will both lie on the preferred path of $P$. Thus, the only access cost is the traversal of the auxiliary splay tree. Since there are only two elements in the working set, they will be at the top of the auxiliary splay tree, identical to a working set on an ordinary splay tree. Thus, the working set access cost will be some constant $A$ associated only with splay tree access. 
    
    Now, let us consider the case in which $x$ and $y$ do not lie on the same root-leaf path of $P$. In this case, each access will require an interleave to switch the preferred child pointers between the two different root-leaf paths to $x$ and $y$. This interleave will require a large constant overhead cost so that the average cost per access is some constant $B > A$. 
\end{proof}

\begin{lemma}
\label{lma:commonpath}
    Let $x$ and $y$ be two distinct keys $x$ and $y$ randomly selected from the universe $\{ 1, 2, \dots, n\}$. The probability that $x$ and $y$ lie on a common root-leaf path in $P$ (i.e. the probability that either $x$ lies in the subtree of $y$ in $P$ or vice-versa) is $2 \lg n / n$ (to leading order).
\end{lemma}
\begin{proof}
    We assume that $n = 2^h - 1$ such that the last level of $P$ is full, where $h$ is the height of $P$. It is easy to see that this result will hold for all other $n$ since we are only concerned with leading order terms. Assume, without loss of generality, that the depth of $x$ $\leq$ depth of $y$ in $P$. The total number of possible pairs $x,y$ is $\binom{n}{2}$, so we focus our efforts on counting the number of pairs such that $x$ is an ancestor of $y$. 

    If $l$ is the depth of $y$, then the number of satisfying $x$ is clearly $l$, the number of ancestors of $y$. Then, the total number of satisfying pairs of keys is simply
    \begin{equation*}
        \sum_{l=0}^{h-1} 2^l l \approx 2^h h \approx n \lg n.
    \end{equation*}
    And, the probability of selecting $x$ and $y$ on the same root-leaf path is
    \begin{equation*}
        \frac{n \lg n}{\binom{n}{2}} \approx \frac{2 n \lg n}{n^2} = \frac{2\lg n}{n}
    \end{equation*}
    as desired. 
\end{proof}

If we let $\rho(n) = 2 \lg n / n$, we can now expect that our data for $k=2$ should fit the model 
\begin{equation}
\label{eq:fitk2}
    \text{average cost} = \rho(n) A + (1 - \rho(n)) B = B - (B-A)\rho(n)
\end{equation}
where $A$ is the `cheap' cost and $B$ is the `expensive' cost based on the cases analyzed in Lemma \ref{lma:multisplay-k2}. When we fit our experimental data to this model, we find a best fit curve as shown in Figure \ref{fig:working_set_multisplay_k}, and constants $A = 10.641$, $B=34.679$. This is consistent with our data since $A$ is approximately equal to the minimum average access cost observed, and $B$ is approximately equal to the maximum average cost observed. So, we have experimental verification of our analysis for the $k=2$ case. 

% Is it multiplicative or additive??
% Since slopes change in right figure above, this means probably multiplicative
For larger $k$, we can see that extending this analysis would involve detailed combinatorial analysis of many cases of keys lying in the same or different preferred paths. Rather than performing that analysis here, we can state that the average cost per access for working set of size $k$ is $\lg k \cdot p(\rho(n))$ where $p$ is some polynomial. Particularly,
\begin{equation*}
    p(\rho(n)) = \sum_{t=0}^T c_t \rho(n)^t.
\end{equation*}
We conjecture that the degree of the polynomial should be $T = k-1$ due to the combinatorial nature of the problem. 

In line with this general analysis, we also find a best-fit curve for our experimental data for $k=4$, using a degree 3 polynomial such that
\begin{equation}
\label{eq:fitk4}
    \text{average cost} = c_0 + c_1 \rho(n) + c_2 \rho(n)^2 + c_3 \rho(n)^3.
\end{equation}
The constants we find for best fit are $c_0 = 59.087$, $c_1 = -175.309$, $c_2 = 305.574$, and $c_3 = -285.435$. 

\subsubsection{Comparison of working set access on all trees}

These sets of data produce a clear hierarchy amongst the three data structures when it comes to the working set access. The splay tree demonstrates its theoretical optimality in practice, with little to no dependence on $n$ on the average cost per access. We can understand why the splay tree performs so optimally by considering the first pass on the working set; since each splay brings the most recently accessed key to the top of the splay tree, after the first pass, all keys are at most depth $k$ from the root of the splay tree. And, once all keys are near the top of the tree, they remain there and quickly occupy the highest $\Theta(\lg k)$ levels of the tree. 

The multi-splay tree has some large overheads which cause it to perform much worse than the splay tree on the working set access, which likely has to do with its requirements to organize its auxiliary trees in a particular fashion in order to achieve $O(\lg \lg n)$-competitiveness. And, finally the tango tree exhibits average cost per access $\Theta(\lg k \lg \lg n)$, which, as described earlier, is likely due to the combination of the use of balanced red-black trees and the structure of the auxiliary trees according to the preferred paths in $P$. 

Finally, we fix $n$ and compare the running times of each of the data structures as $k$ varies to understand their numerical relative performance. As seen in Fig. \ref{fig:working_set}, The splay outperforms the multi-splay tree by a factor of about 30,   and the multi-splay tree outperforms the tango tree by a factor of about 5. So, as with the other access sequences, the splay tree is most practically efficient. 

\begin{figure}[ht]
    \centering
    \includegraphics[height=2.5in]{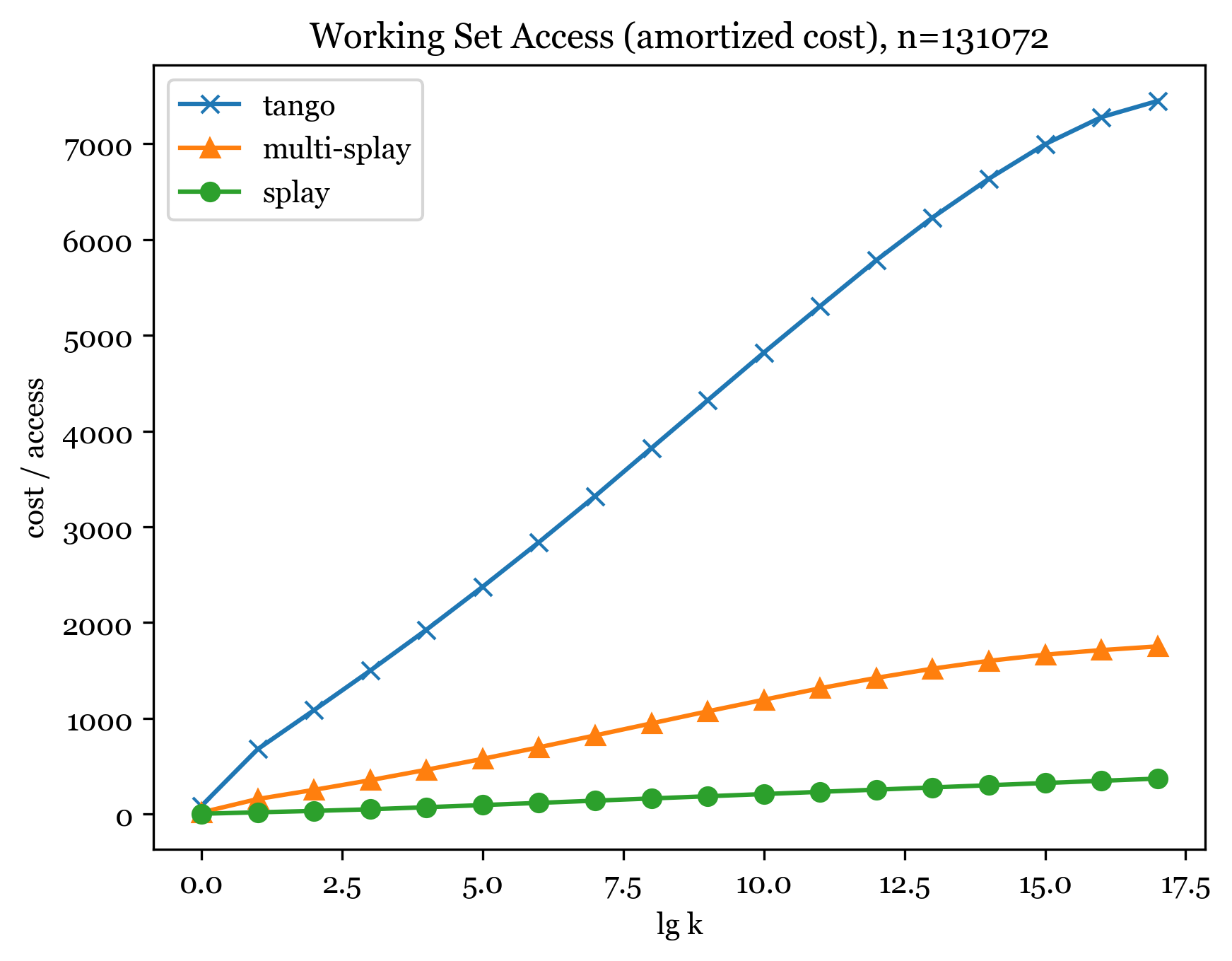}
    \caption{Results of running a working set access sequence on the tango tree, multi-splay tree, and splay tree as $n=131072$ is fixed and $k$ varies.}
    \label{fig:working_set}
\end{figure}

% TODO: add figures
\subsection{Unified Property}
The unified property has not been proven to hold for multi-splay trees or splay trees. Therefore, experimental results are valuable in solving this problem.

The results for the unified property were similar to those for the working set property, which could indicate that splay trees and multi-splay trees hold the unified property. Because of this, we omit the graphs for the tango tree and multi-splay as they are very similar.

However, one surprising finding was that the cost per access increases sub-linearly with respect to $\lg k$ for the splay tree, as shown in Fig. \ref{fig:unified_splay}. This suggests that there may be another property related to the specific access sequence used in the experiment. However, increasing the number of elements with a fixed value of $k$ did not affect the cost per access (similar to as seen for the working set property).

\begin{figure}[htb]
    \centering
    \includegraphics[height=2.5in]{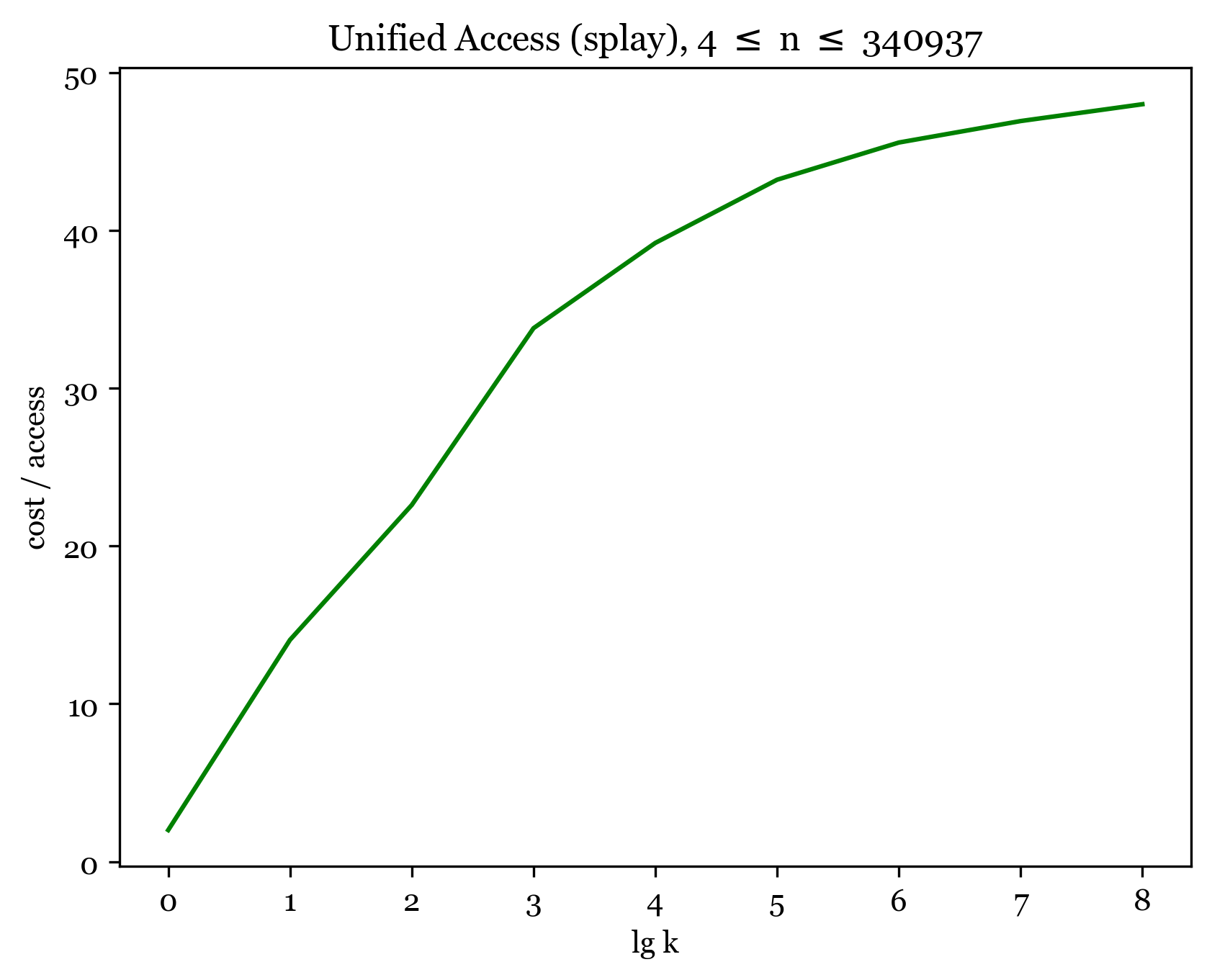}  
    \caption{
    % Results of running a unified set access sequence on the splay tree, as $n$, the number of keys, and $k$, the size of the working set vary. 
    Results of running a unified set access sequence on the splay tree as $k$ the size of the working set varies. 
    In the second plot, we have averaged the access costs for $4 \leq n \leq 340937$ (as appropriate for each value of $k$) to provide clarity since the costs were nearly identical across different values of $n$.}
    \label{fig:unified_splay}
\end{figure}

\section{Conclusion}

Overall, we found that in many practical cases, the splay tree outperforms other binary search tree algorithms. This is despite the fact that tango trees and multi-splay trees are proven to be $O(\lg \lg n)$-competitive while no such bound is known for splay trees (beyond the trivial $O(\lg n)$-competitiveness). For, the sequences we tested, splay trees seems to suffer from little overhead compared to the other trees, especially on large access sequences. 

Our experiments revealed that the tango tree, in practice, seems to have a multiplicative $\Theta(\lg \lg n)$ overhead on many access sequences (compared to the optimal running time). We prove that this is the case for the sequential access, and our detailed proof (which is, as far as we know, the first of its kind), provides insight into why this seems to be true. Specifically, the use of red-black trees as auxiliary trees seems to force this overhead to be present on many accesses. 

This would give hope for the multi-splay tree, which uses splay trees instead as auxiliary trees. And, while multi-splay trees are proven to be asymptotically optimal on many access sequences, our experiments show that the constant terms on multi-splay running times are far from optimal in absolute terms. Due to their organization based on preferred paths, multi-splay trees seem to suffer from larger overheads and constant factors, resulting in slow downs of one or more orders of magnitude compared to the splay tree. This indicates that the current approach to competitive binary search trees (i.e. via the interleave bound) may not be optimal for achieving practical performance.

Our last conclusion is that splay trees and multi-splay have shown experimental hints of following the unified property. This only gives more plausibility to the belief that splay trees and multi-splay trees are dynamically optimal. 

Overall, this work contributes a first-of-its-kind extensive experimental implementation and analysis of competitive binary search trees, and uses this to develop novel insights about the behavior of these trees.

\section{Future Work}
This study has provided the first experimental testing, to the best of our knowledge, of various BST properties, such as the working set property and unified property. The results were promising with respect to splay and multi-splay optimality. However, there are still other BST properties that need to be experimentally tested.

Furthermore, our study has focused on a limited number of BSTs. A larger research can be conducted on a larger set of known BSTs to further explore their properties and relative performance.

Moreover, there are interesting behaviors that can be observed in relation to the formulated access sequence and its relationship to the unified property. One area for further investigation would be to experiment with different types of access sequences with a fixed unified bound, or to explore the non-linear trend observed with increasing $\lg k$.

Finally, while our choice to use the BST computation model for measuring access costs is useful for comparing relative performance and understanding asymptotic trends, it may be useful to also measure real running times in order to understand the practical effects of e.g. cache, recursive overhead, etc. on various BST implementations. 

%%
%% Bibliography
%%

%% Please use bibtex, 

\bibliography{main}

\appendix

% TODO: shorten this section, and maybe add some more about random accesses 
\section{Proof of tango tree running time on sequential access}
\label{sec:tango-worst-case}

Here, we present a more formal proof of the worst-case competitive running time of the tango tree on the sequential access. Consider a tango tree $T$ which is initialized with no preferred children (i.e. each auxiliary tree has size 1) and the associated complete binary tree $P$. Henceforth we will refer to the preferred path of the root of $P$ as simply the preferred path of $P$ for brevity. 

\begin{lemma}
\label{lma:pref-length}
    After the first access in the sequential access, the preferred path of $P$ always has extends from the root of $P$ to a leaf of $P$ (and thus has length $\Theta(\lg n)$).
\end{lemma}
\begin{proof}
    We can consider two cases: immediately after an access to a leaf of $P$, and immediately after an access to a non-leaf node of $P$. In the first case, the preferred path of $P$ must contain the just-accessed leaf node, immediately satisfying the claim.

    In the second case, we consider the preferred path immediately after accessing a non-leaf node $y$ of $P$. Since $y$ is not a leaf node, $y$ must have a left child; this is because the definition of a complete binary tree requires that every node of $P$ either has no children, two children, or only a left child.  Let $x$ be the rightmost descendant of the left child of $y$ (which must exist and be distinct from $y$). Clearly, $x$ is the predecessor to $y$ in $P$, so we must have accessed $x$ immediately before accessing $y$.
    
    This means that before accessing $y$, $x$ was in the preferred path of $P$, and since $y$ is an ancestor of $x$, $y$ was also in the preferred path of $P$. Additionally, $y$'s preferred child was set to its left child since $x \in L(y)$. Since $y \in L(y)$, accessing $y$ will not change the preferred path of $P$, and so the preferred path of $P$ after accessing $y$ must contain $x$, as desired.
\end{proof}

\begin{corollary}
\label{cor:leaf-access-order}
    During the sequential access, for all leaf nodes $x$ of $P$, $x$ is not added to the preferred path of $P$ until $x$ itself is accessed. 
\end{corollary}
\begin{proof}
   We showed above that the preferred path is not updated when accessing any non-leaf node of $P$. So, the preferred path of $P$ always contains the most recently accessed leaf of $P$. Clearly, then, a leaf $x$ of $P$ is only added to the preferred path of $P$ when $x$ is accessed (and it could not have been added earlier).
\end{proof}

\begin{lemma}
\label{lma:leaf-access-time}
    During the sequential access, for all leaf nodes $x$ of $P$, the running time of the access to $x$ is $\Omega(1 + \lg \lg n)$.
\end{lemma}
\begin{proof}
    Note that for the first access to $x=1$, we have to traverse through $\Theta(\lg n)$ distinct auxiliary trees of size 1, so the access time is $\Omega(\lg n) \subset \Omega(1 + \lg \lg n)$. 
    
    Right before accessing any other leaf node $x$, $x$ is not part of the preferred path of $P$, from Corollary \ref{cor:leaf-access-order}. Thus, there is at least one interleave during the access to $x$. In particular, let us we consider the auxiliary tree $\AT(P)$ in $T$ representing the preferred path of $P$. During the access to $x$, we have to take some non-preferred child from $\AT(P)$, into a different auxiliary tree. 
    
    Recall that the auxiliary trees are represented using red-black trees, and furthermore that a red-black tree with $k$ nodes satisfies the property that all of its `null' children have depth $\Theta(\lg k)$. Here, `null' children are the children of leaves and the non-real children of nodes with only one child. In the case of the tango tree, all non-preferred children are `null' children of the auxiliary tree. So, to follow a non-preferred child from an auxiliary tree of size $k$ requires navigating to a `null' child at depth $\Theta(\lg k)$ from the root of the auxiliary tree. 

    Before accessing $x$, the auxiliary tree at the root of $T$ is $\AT(P)$ where $|\AT(P)| = \Theta(\lg n)$ from Lemma \ref{lma:pref-length}. And, we know $x \notin \AT(P)$, so some non-preferred child must be followed to access $x$. Navigating to this `null' non-preferred child takes $\Theta(1 + \lg \lg n)$ time due to the red-black structure of the auxiliary tree as noted above. And, so the access to node $x$ takes $\Omega(1 + \lg \lg n)$ time since the total access cost includes this interleave cost.
\end{proof}

\begin{theorem}
    The running time of the tango tree algorithm on the sequential access $X =  1, 2, \dots, n $ is $\Theta(n (1 + \lg \lg n))$.
\end{theorem}
\begin{proof}
    Since the optimal offline time for the sequential access is $O(n)$, from Corollary \ref{cor:competitive}, we immediately know that the running time of tango tree on the sequential access is $O(n (1 + \lg \lg n))$.
    
    From Lemma \ref{lma:leaf-access-time}, we know that each access to a leaf node of $P$ takes time $\Omega(1 + \lg \lg n)$. Since the number of leaf nodes of $P$ is $\geq n/2 = \Theta(n)$ (which follows from the definition of the complete binary tree), the total running time of the tango tree algorithm on the sequential access must be $\Omega(n (1 + \lg \lg n))$.
\end{proof}

So, we have proven the existence of access sequences for which the competitive bound on the tango tree running time is tight, indicating that the bound cannot be improved.

\end{document}